\documentclass{article}
\usepackage{spconf}
\usepackage{amsmath,graphicx}
\usepackage{algorithm, algorithmic}
\usepackage{amsfonts}
\usepackage{float}  
\usepackage{lipsum}
\usepackage{amsthm}
\usepackage{subfigure}
\usepackage{cite}
\usepackage{diagbox}
\newtheorem{theorem}{Theorem}

\newtheorem{lemma}{Lemma}
\makeatletter
\newenvironment{breakablealgorithm}
{
	\begin{center}
		\refstepcounter{algorithm}
		\hrule height.8pt depth0pt \kern2pt
		\renewcommand{\caption}[2][\relax]{
			{\raggedright\textbf{\ALG@name~\thealgorithm} ##2\par}%
			\ifx\relax##1\relax 
			\addcontentsline{loa}{algorithm}{\protect\numberline{\thealgorithm}##2}%
			\else 
			\addcontentsline{loa}{algorithm}{\protect\numberline{\thealgorithm}##1}%
			\fi
			\kern2pt\hrule\kern2pt
		}
	}{
		\kern2pt\hrule\relax
	\end{center}
}
\makeatother

\title{recovery of graph signals from sign measurements}
\name{Wenwei Liu$^\star$, Hui Feng$^\ast$, Kaixuan Wang$^\star$, Feng Ji$^\dagger$, Bo Hu$^\ast$}
\address{$\star,\ast$ Research Center of Smart Networks and Systems, School of Information Science and Engineering\\
Fudan University, Shanghai 200433, China\\
$^\dagger$ School of Electrical and Electronic Engineering\\
Nanyang Technological University, Singapore\\
$^\star\lbrace$wwliu21,kaixuanwang21$\rbrace$@m.fudan.edu.cn, $^\ast\lbrace$hfeng, bohu$\rbrace$@fudan.edu.cn, $^\dagger$jifeng@ntu.edu.sg}

\begin{document}

\maketitle
\begin{abstract}
Sampling and interpolation have been extensively studied, in order to reconstruct or estimate the entire graph signal from the signal values on a subset of
vertexes, of which most achievements are about continuous signals.
While in a lot of signal processing tasks, signals are not fully observed, and only the signs of signals are available, for example a rating system
may only provide several simple options.
In this paper, the reconstruction of band-limited graph signals based on sign sampling is
discussed and a greedy sampling strategy is proposed. The simulation experiments are
presented, and the greedy sampling algorithm is compared with random sampling
algorithm, which verify the validity of the proposed approach.
\end{abstract}
\section{Introduction}
\label{sec:intro}

In general, many types of data in life can be modeled as graph signals, such as transportation networks, social networks, etc\cite{stankovic2019understanding,zhou2004regularization,zhang2020deep}.
Graph signal processing (GSP) has been one of the emerging fields in the signal processing community in recent years\cite{shuman2013emerging,stankovic2020data}.
Based on current research, GSP has a wide range of application scenarios in semantic segmentation, behavior
recognition, community discovery, traffic prediction and other aspects\cite{xia2021graph}.
Based on the graph Fourier Transform (GFT), techniques analogous to those in the classical Fourier theory have been developed in the GSP setting\cite{cvetkovic2009applications}. 
Common graph signals in life often show bandlimited characteristics, i.e. signal values are similar on similar vertexes\cite{anis2016efficient}.
For instance, in social networks, people who are closely connected tend to show similar preferences.
This paper focuses on such signals.

It is a fundamental problem of GSP to recover graph signals from limited samples with noise \cite{tanaka2020sampling,girault2020graph}. 
At present, there are two main sampling schemes: deterministic selection and random selection\cite{tanaka2020sampling,tsitsvero2016signals}. The deterministic approach is trying to 
find a sampling set that minimizes a preset loss function. While the random approach usually calculates the sampling probability distribution of each node according to its importance.
In practice, random selection is of low computational load, but may need more samples than deterministic selection to achieve equivalent sampling effect for subsequent reconstruction.  

In many cases, the data we're dealing with is enormous as well as hard to get accurate and continuous values. 
For instance, in a goods ranking system, we may only see the people's simple evaluations for goods such as ``like",``dislike",``indifferent". 
In this scenario, getting an exact score for each item is often a hassle, whereas it may be easier to find out it is just rated as ``like" or ``dislike", which can be viewd as $1$ and $-1$ respectively.
Such ``like" and ``dislike" are the sign information we work on.  

As a extremely simple information on a graph signal, acquiring the signs of signal value on nodes results in small computation load, because it is a low bit quantization of the graph signal\cite{jacques2013robust}.
Due to this advantage of sign information, we try to recover the original graph signal from the limited samples with the help of such coarse information. 



In this paper, we investigate the problem of reconstructing band-limited graph signals from the signs of signal values on a subset of vertexes.
we propose an effective signal sampling scheme and a corresponding signal
reconstruction algorithm based on sign information.
Furthermore, the characteristics and effects of sampling and recovery algorithm are analyzed theoretically.
Finally, we present some simulation experiments on different data sets to test the performance of the proposed algorithm, and the results verify the validity of our work.



\section{Model}
\label{sec:model}

A connected undirected graph $\mathcal{G}$ with $\mathit{N}$ vertexes can be represented as: $\mathcal{G}=\{\mathcal{V},\mathcal{E},\boldsymbol{W}\}$, where
$\mathcal{V}$ is the set of vertexes, $\mathcal{E}$ is the set of edges, and $\boldsymbol{W}\in \mathbb{R}^{\mathit{N}\times\mathit{N}}$ is the weighted adjacency
matrix with elements $\mathit{W}_\mathit{ij} >0$ if $(\mathit{i},\mathit{j})\in\mathcal{E}$, otherwise $\mathit{W}_\mathit{ij} =0$.
The degree of vertex $\mathit{i}$ is $\mathit{d_i} = \sum_{j}\mathit{W_{ij}}$, and the degree matrix $\boldsymbol{D}$ is a $\mathit{N}\times\mathit{N}$ diagonal whose $i$th diagonal element is $\mathit{d_i}$.
The graph Laplacian $\boldsymbol{L}$ is defined as: $\boldsymbol{L} = \boldsymbol{D}-\boldsymbol{W}$. Obviously, it is real symmetric, with eigendecomposition
$\boldsymbol{L}=\boldsymbol{U}\boldsymbol{\Lambda}\boldsymbol{U}^\mathit{T}$, where $\boldsymbol{\Lambda}$ is diagonal matrix of eigenvalues $0=\lambda_1\le\lambda_2\le\dots\le\lambda_N$
and $\boldsymbol{U}=[\boldsymbol{u}_1,...,\boldsymbol{u}_N]$ is an orthonormal matrix containing the corresponding eigenvectors.

A graph signal $\boldsymbol{x}:=\mathcal{V}\mapsto\mathbb{R}$ is a vector defined on graph $\mathcal{G}$. Its GFT can be written as: $\hat{\boldsymbol{x}}=\boldsymbol{U}^\mathit{T}\boldsymbol{x}$.
Each component of $\boldsymbol{x}$ is the corresponding frequency domain coefficient. Let $f_L,f_U$ be positive integers, if a graph signal satisfies: $\hat{\boldsymbol{x}}_\mathit{k}=0$ for all $\mathit{k}\not\in[\mathit{f_L},\mathit{f_U}]$, then it is called band-limited
with passband $[\mathit{f_L},\mathit{f_U}]$, bandwidth $B=\mathit{f_U}-\mathit{f_L}+1$. For such signals, they are in the following constraint space:
\begin{equation}
  \mathit{C}_b=\left\{\boldsymbol{x}\in\mathbb{R}^\mathit{N}|\,\forall \mathit{i}\not\in[\mathit{f_L},\mathit{f_U}], \boldsymbol{u}_i^T\boldsymbol{x}=0\right\}
\end{equation}
According to the above formula, it is easy to see that $C_b$ is a closed convex cone in $\mathbb{R}^\mathit{N}$ \cite{boyd2004convex}.

Consider obtaining $\mathit{M}$ samples on a subset $\mathcal{V}'$ of $\mathcal{V}$, i.e. $|\mathcal{V}'|=\mathit{M}$, then the sampling process can be described as:
\begin{equation}\nonumber
  \boldsymbol{y} = \boldsymbol{\Psi}_\mathit{v}\boldsymbol{x}
\end{equation}
where $\boldsymbol{\Psi}_\mathit{v}\in\mathbb{R}^\mathit{M\times N}$ is vertex-domain sampling matrix, with elements $(\boldsymbol{\Psi}_\mathit{v})_\mathit{ij}=1$ if the $\mathit{i}$th sample is
$\mathit{j}$, and $0$ otherwise.

The sign information at a vertex comes from a sign operation of the signal at the corresponding vertex. For example, the sign on vertex $\mathit{i}$ is $sign(\boldsymbol{x}_\mathit{i})$, where
\begin{equation}
  sign(x)=
  \left\{
    \begin{array}{lr}
      -1      & x<0 \\
      1      & x>0 \\
      0                 & otherwise 
    \end{array}
  \right.
\end{equation}
The sampling process of sign information on vertexes can be described as:
\begin{equation}
  \boldsymbol{S}_\mathit{x}=sign\left(\boldsymbol{\Psi}_\mathit{v}\boldsymbol{x}\right)
\end{equation}

For a signal with sign $\boldsymbol{S}_\mathit{x}$ on vertexes, from (3) we can derive that
the signal should be confined to a constraint space whose closed convex hull is:
\begin{equation}
  \label{sign space}
  \mathit{C_v}=\mathop{\cap}\limits_{\mathit{i}=1}^{\mathit{M}}\left\{\boldsymbol{x}\in\mathbb{R}^\mathit{N}|\,\boldsymbol{a}_\mathit{i}^\mathit{T}\boldsymbol{x}\leq 0\right\}
\end{equation}
where $\boldsymbol{a}_\mathit{i}$ is a all zero vector except that its $\mathit{j}$th component is $\pm 1$, which depends on $\left(\boldsymbol{S}_\mathit{x}\right)_\mathit{i}$.
If $\left(\boldsymbol{S}_\mathit{x}\right)_\mathit{i} = 0$, then ``$\leq$'' in (\ref{sign space}) becomes ``$=$".
It is not hard to prove that $\mathit{C_v}$ is also a closed convex cone in $\mathbb{R}^\mathit{N}$\cite{boyd2004convex}.

\section{Reconstruction Algorithm}
\label{sec:POCS}
Based on the model in the previous section, the graph signal is recovered by continuous projection
onto the convex constraint spaces in this section\cite{theodoridis2010adaptive}.

Through sign information in the sampling set, the original band-limited graph signal lies in space: $C=\mathit{C_b}\cap\mathit{C_v}$, which is likewsie a closed convex cone.
Due to the fact that $C_b$ and $C_v$ are both closed convex cones, we can define the projection operators of them\cite{theodoridis2010adaptive}.

Projecting any signal onto $\mathit{C}_b$ requires only one band-pass filtering operation:
\begin{equation}
  \boldsymbol{P}_b=\boldsymbol{U}\boldsymbol{\Gamma}\boldsymbol{U}^T
  \label{Pb}
\end{equation}
where $\boldsymbol{\Gamma}$ is a diagonal matrix in which the diagonal elements are $1$ inside the passband and $0$ outside the passband.
 
The projection operation onto $\mathit{C_v}$ can be defined as:
\begin{equation}
  \left(\boldsymbol{P}_\mathit{v}\boldsymbol{x}\right)_\mathit{j}:=
  \left\{
    \begin{array}{lr}
      0      & j\in \mathcal{V'},sign \left(\mathit{x_j}\right)\not=\left(\boldsymbol{S}_\mathit{x}\right)_\mathit{i} \\
      \mathit{x_j}                 & otherwise 
    \end{array}
  \right.
  \label{Pv}
\end{equation}

According to the definition of constraint space and projection operator, 
the signal can be reconstructed by continuously projecting onto these convex sets. 
Specifically, for any arbitary initial signal $\boldsymbol{x}_0$, the applied algorithm is a simple iterative process:  
\begin{equation}
  \boldsymbol{x}_{\mathit{n}+1}=\boldsymbol{P}_\mathit{b}\boldsymbol{P}_\mathit{v}\boldsymbol{x}_\mathit{n}
  \label{iterformula}
\end{equation}

In terms of the settings of the reconstruction iterative algorithm, we can verify that each projection is a firmly non-expansive operation. At the same time, 
the two constraint spaces are very special and conform to the characteristic of bounded linear regularity.

The following theorem can be derived based on these conditions.
\begin{theorem}
  The iterative sequence $\left\{\boldsymbol{x}_\mathit{n}\right\}$ converges to some point $\boldsymbol{x}^*$ in $C$, and the convergence
  rate is independent of the selection of the initial point $\boldsymbol{x}_0$.
\end{theorem}

\begin{proof}
  See Appendix A for details.  
\end{proof}

\section{Design of Sampling Set}
\label{sec:sampling}
The previous section is about methods for recovering signals from sampled information.
This section explores how to optimize recovery.

Since the recovery sequence of the graph signal will converge to the feasible region, an attempt
to optimize the feasible region can be made to improve the recovery performance.
As we know, in the whole process of sampling and recovery, the feasible region of the signal depends on the sampling
results, which is manageable for us, we can adjust the feasible region by adjusting the sample set.
To sum up, we aim to design a sampling set with a better feasible region.

\subsection{Feasible Region Analysis}
\label{ssec:feasible}
From the previous discussion, we can already make it clear that the feasible region is a finite-dimensional closed convex cone.
If we want to make it smaller, it is necessary to determine a metric for its size.

Any vector in a closed convex  cone can be represented linearly by extreme vectors with non-negative coefficients\cite{barker1973lattice}:
\begin{equation}
  C=\left\{\sum_{i=1}^r \omega_i \boldsymbol{\varphi} _i |\, \omega_i \geq 0\right\}
\end{equation}
In our paper, the convex cone is finite-dimensional, i.e. the set of extreme vectors is finite.

Denote $\boldsymbol{U}_B=[\boldsymbol{u}_{f_L},...,\boldsymbol{u}_{f_u}]$, then since the columns of $\boldsymbol{U}_B$ are orthonormal, 
for any two signals $\boldsymbol{\beta}_1$, $\boldsymbol{\beta}_2$ in $C$ with coordinates $\boldsymbol{\alpha}_1$, $\boldsymbol{\alpha}_2$, we can get 
$\langle\boldsymbol{\beta}_1,\boldsymbol{\beta}_2\rangle=\langle\boldsymbol{U}_B\boldsymbol{\alpha}_1,\boldsymbol{U}_B\boldsymbol{\alpha}_2\rangle=\langle\boldsymbol{\alpha}_1,\boldsymbol{\alpha}_2\rangle$.
So in other words, denoting the feasible region of coordinates as $\hat{C}$, we can transform the exploration of the $C$ into the exploration of $\hat{C}$, because it is of one-to-one correspondence between vectors in $C$ and $\hat{C}$.
According to (\ref{sign space}), $\hat{C}$ can be written as:
\begin{equation}
  \hat{C} = \mathop{\cap}\limits_{\mathit{i}=1}^{\mathit{M}}\left\{\boldsymbol{x}\in\mathbb{R}^\mathit{B}|\,\boldsymbol{a}_\mathit{i}^\mathit{T}\boldsymbol{U}_B\boldsymbol{x}\leq 0\right\}
\end{equation}

Since the recovery signal is expected to be as close to the original signal as possible, that is, the angle between them is as small as possible, hence 
we try to use the biggest angle among extreme vectors as the metric. Let $Z$ be the set of normalized extreme vectors of $\hat{C}$, then the metric for its size can be represented as:
\begin{equation}
  \label{theta}
  \theta=\mathop{min}\limits_{\boldsymbol{\gamma},\boldsymbol{\mu} \in Z}\langle\boldsymbol{\gamma},\boldsymbol{\mu}\rangle
\end{equation}

\subsection{The Corresponding Optimization Problem}
\label{ssec:optimization}
Since the feasible region in the sampling process needs to satisfy the sign constraints under the current sampling result, 
an intuitive idea is that each additional sampling constraint can minimize the feasible region to the greatest extent.
On the grounds of this idea, we can express the problem in mathematical form. Assume the current sampling set is $\mathcal{S}$, next sample is $\xi$ after sampling $s$ times, 
$\boldsymbol{\gamma},\boldsymbol{\mu}$ are the coordinates of any two extreme vectors of $\hat{C}$ under orthonormal basis $\boldsymbol{U}_B$.
Then our sampling process can be described as: 
\begin{equation*}
  \begin{aligned} \label{P}
    &\mathop{max}\limits_{\xi}\left(\mathop{min}\limits_{\boldsymbol{\gamma,\mu}\in Z}\quad\frac{\langle\boldsymbol{\gamma,\mu}\rangle}{\left\|\boldsymbol{\gamma}\right\|\cdot\left\|\boldsymbol{\mu}\right\|}\right)\\
    \text{s.t.}\quad &sign\left[\mathop{\sum}\limits_{i=f_L}^{f_U}\left(\gamma_i\boldsymbol{u}_i\right)_j\right]=sign\left(\boldsymbol{x}_j\right)\\
      &sign\left[\mathop{\sum}\limits_{i=f_L}^{f_U}\left(\mu_i\boldsymbol{u}_i\right)_j\right]=sign\left(\boldsymbol{x}_j\right)\\
  \end{aligned}
\end{equation*}
where $j\in\mathcal{S}\cup{\{\xi\}}$.

\subsection{Greedy Sampling Algorithm}
\label{ssec:greedy}
According to the above optimization problems, the corresponding sampling algorithm can be designed.
It should be noted that there is no extreme vector until the number of samples reaches $B-1$. Therefore, we need to determine the first $B-1$ samples according to other strategies.
Here, we sort the norms of the row vectors of $\boldsymbol{U}_B$ in descending order, and select the first $B-1$ indices as the first $B-1$ samples.
After that, we can continue sampling according to the principle of reducing the maximum angle of extreme vectors.
But if we compute all the extreme vectors for each unsampled vertex every time, it's of very high complexity.

To address this issue, we propose a greedy sampling algorithm.
As we know, for every additional sample, we're adding a constraint to our current feasible region. For example, if we choose $\xi$ as the next sample, of which the sign is negative, in that way
the new feasible region becomes $\hat{C}\cap \{\boldsymbol{x} |\,\boldsymbol{U}_B(\xi,:)\boldsymbol{x} \leq 0\}$, where $\boldsymbol{U}_B(\xi,:)$ is the $\xi$th row vector of $\boldsymbol{U}_B$.
The boundary of the new constraint is a hyperplane $\{\boldsymbol{x} |\,\boldsymbol{U}_B(\xi,:)\boldsymbol{x} = 0\}$. To better illustrate, refer to Figure \ref{C}.
\begin{figure}[h]
  \centering
  \includegraphics[scale=0.3]{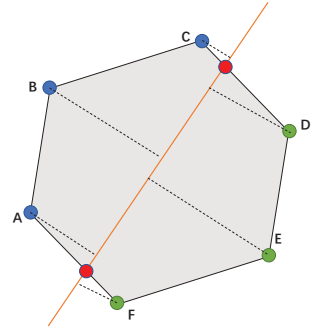}
  \caption{An example of a feasible region in sampling}
  \label{C}
\end{figure}
In Figure \ref{C}, the gray polygon area is a cross section of this cone, i.e. the feasible region, in down direction towards the origin. The solid orange line stands for the hyperplane $\{\boldsymbol{x} |\,\boldsymbol{U}_B(\xi,:)\boldsymbol{x} = 0\}$.
The blue and green solid dots represent the extreme vectors in the current feasible region, and the red solid dots represent the new extreme vectors associated with the hyperplane.
Without loss of generality, assume the left halfspace where the blue dots lie is $\{\boldsymbol{x} |\,\boldsymbol{U}_B(\xi,:)\boldsymbol{x} \leq 0\}$, which means if we choose $\xi$ as the next sample, then the new feasible region is
the part of the current feasible region to the left of the hyperplane. On condition that the sign of the signal on $\xi$ is positive, the part on the right side of the hyperplane is the new feasible region instead.
All unsampled vertexes correspond to a hyperplane, our target is to find a suitable unsampled vertex that its hyperplane can reduce the maximum angle of the extreme vectors.

For the extreme vectors with the biggest angle $\boldsymbol{t}_1,\boldsymbol{t}_2$ in the current feasible region, for example, $A$ and $D$ in Figure \ref{C}, if $\boldsymbol{t}_1$ and $\boldsymbol{t}_1$ are not separated by the new hyperplane of a unsampled vertex, then we would not choose it as the next sample, 
because if this pair of vectors are still in the new feasible region probably, the new feasible region might be as big as the previous one. We call such sampling invalid.

How to evaluate all the valid samples is the following question. Here we can see all the valid hyperplanes divide the feasible region into two parts. If these two parts are divided as evenly as possible, 
then whatever the sign is, the new feasible would not be too bad. In other words, if these two parts are greatly different from each other, there is a high risk that the new feasible region would be almost as large as the previous one.
As a result, we want all the extreme vectors to lie roughly ``equidistant" on both sides of the hyperplane under consideration. Refer to Figure \ref{C}, the dotted line segment stands for the distance from the corresponding extreme vector to the hyperplane.
If the sum of these distances on both sides are numerically close, as thus the feasible region is roughly ``cut" in half. Based on these analyses, the greedy algorithm can be presented as: 
\begin{breakablealgorithm} 
	\renewcommand{\algorithmicrequire}{\textbf{Input:}}
	\renewcommand{\algorithmicensure}{\textbf{Output:}}
	\caption{Greedy Sampling Algorithm} 
	\label{alg:1} 
	\begin{algorithmic}[1]
		\REQUIRE Sampling budget $M$, $\boldsymbol{U}_B$, bandwidth $B$, graph topology $\mathcal{G}$
		\ENSURE sampling set $\mathcal{S}$
    \STATE $\mathcal{S}\gets\{\}$,$r\gets-\infty$
    \STATE sort the norms of row vectors of $\boldsymbol{U}_B$ in descending order, get the first $B-1$ indices $\mathcal{S}_0$, $\mathcal{S} \gets \mathcal{S}_0$ 
    \FOR {$i$ in $\mathcal{G.V}\backslash\mathcal{S}$}
    \STATE calculate each $Z$ for $x_i > 0,x_i < 0,x_i = 0$, and the respective $\theta_1,\theta_2,\theta_3$ using (\ref{theta})
    \STATE $\theta \gets min\{\theta_1,\theta_2,\theta_3\}$
    \IF {$\theta > r$}
    \STATE $r \gets \theta$,$j \gets i$
    \ENDIF
    \ENDFOR
    \STATE $\mathcal{S} \gets \mathcal{S} \cup \{j\}$
    \STATE determine $\boldsymbol{t}_1$,$\boldsymbol{t_2}$ according to the acquired $sign(x_j)$
    \STATE $s \gets B+1$,$d \gets +\infty$
    \WHILE{$s < M+1$}
    \FOR {$i$ in $\mathcal{G.V}\backslash\mathcal{S}$}
    \IF {$\left[\boldsymbol{U}_B(i,:)\boldsymbol{t}_1\right]\cdot\left[\boldsymbol{U}_B(i,:)\boldsymbol{t}_2\right]<0$}
    \STATE calculate the sum of the distances $d_1,d_2$ of the extreme vectors on both sides to the hyperplane $\{\boldsymbol{x} |\,\boldsymbol{U}_B(i,:)\boldsymbol{x} = 0\}$
    \IF{$|d_1-d_2|<d$}
    \STATE $d \gets|d_1-d_2|$,$j\gets i$
    \ENDIF
    \ENDIF
    \ENDFOR
    \STATE $\mathcal{S} \gets \mathcal{S} \cup \{j\}$, $s \gets s+1$,$d \gets +\infty$
    \STATE update extreme vectors
		\ENDWHILE 
	\end{algorithmic} 
\end{breakablealgorithm}

\section{Simulation}
\label{sec:experiment}
In this section, we investigate the performance of the greedy sampling algorithm proposed in the previous section through simulation experiments.

In the following experiment, we first apply the  Algorithm 2 and some other sampling methods to get the corresponding sampling sets, and next we choose $K$ initial signals to recover the original signal using (\ref{iterformula}).
Then we compare the reconstruction quality by the average error in angle defined as:
\begin{equation}
  \delta = \frac{1}{K}\mathop{\sum}\limits_{i=1}^K \arccos\langle x,x_i^* \rangle
\end{equation}
where $x_i^*$ stands for the normalized recovery signal of the $i$th initial signal ($1\leq i \leq K$). The larger $\delta$ is, the bigger the error of recovery is.



Consider a band-limited graph signal on a sensor graph generated from gspbox in Matlab, which is shown in Figure \ref{graph}\cite{perraudin2014gspbox}.The parameters of the graph topology are: $N=40$, $|\mathcal{E}|=153$, $f_L=29$, $f_U=35$.
\begin{figure}[h]
  \centering
  \subfigure[A band-limited graph signal of unit amplitude]
  {
    \label{grapha}
    \hspace{-25pt}
      \begin{minipage}[h]{.45\linewidth}
          \centering
          \includegraphics[scale=0.32]{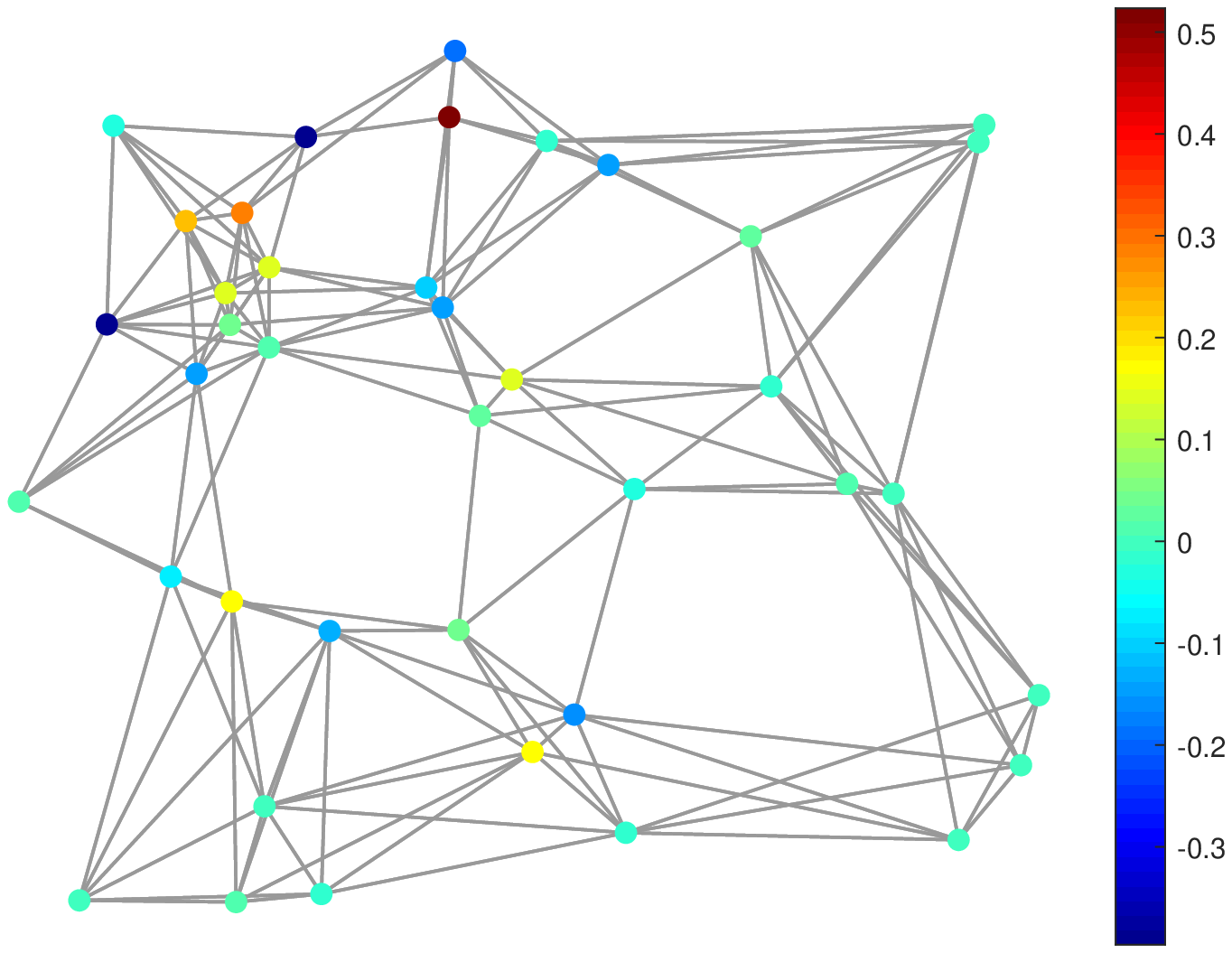}
      \end{minipage}
  }
  \subfigure[Sign information of (a)]
  {
    \label{graphb}
     \begin{minipage}[h]{.45\linewidth}
          \centering
          \includegraphics[scale=0.32]{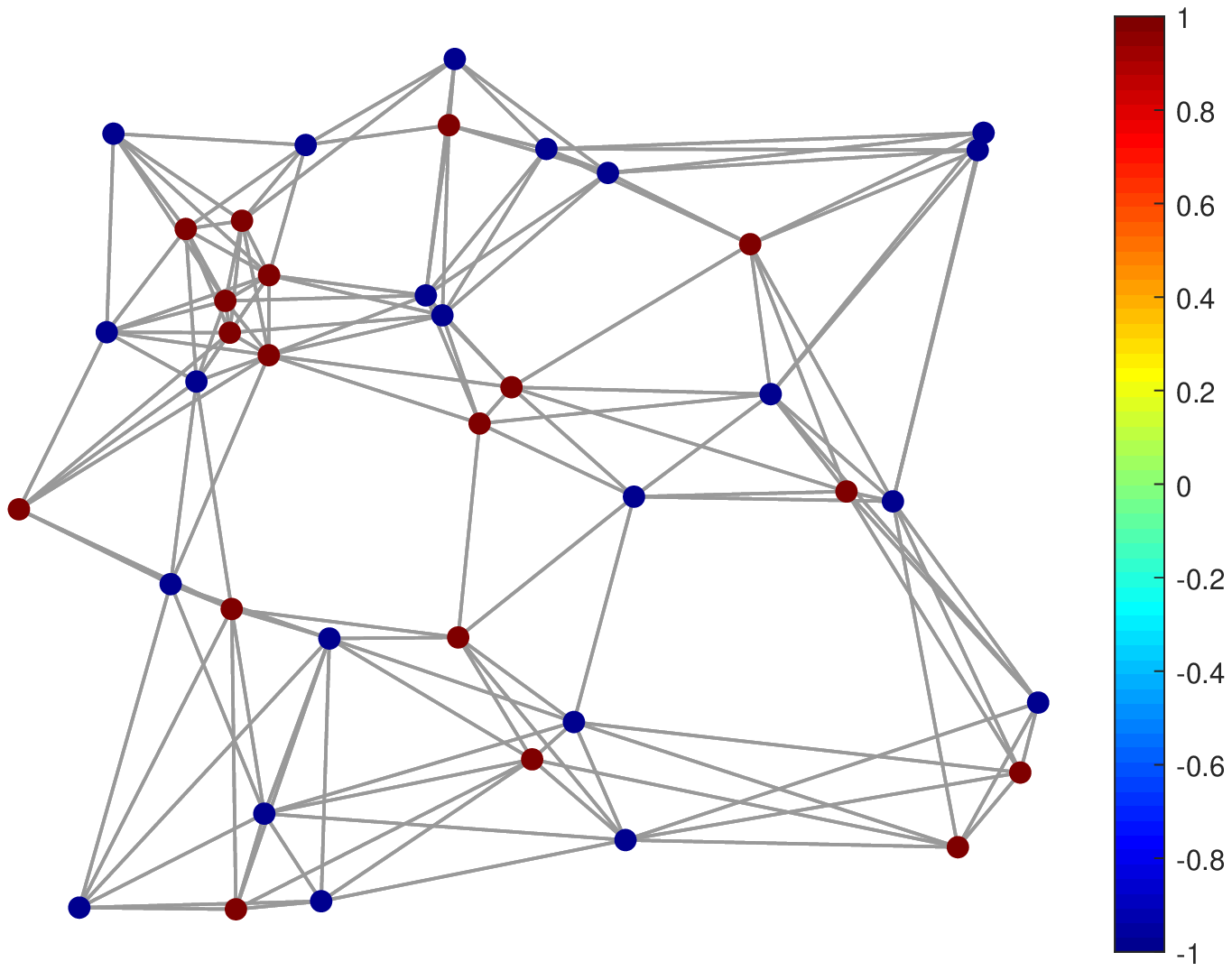}
      \end{minipage}
  }
  \caption{A band-limited graph signal of unit amplitude and its sign information on vertexes}
  \label{graph}
\end{figure}

We applied our sampling algorithm and random sampling algorithm 50 times to obtain their respective sample sets, 
and then calculate the recovery error of 10000 iterations under 50 initial random signals for each sampling set. 

We compared the average error in angle of greedy sampling and random sampling at different sampling rate, as presented in the Figure \ref{comparsion}\subref{vs}, where
the dotted pink line stands for the average error in angle for full sampling.

\begin{figure}[h]
  \centering
  \subfigure[]
  {
	  \label{vs}
	  \hspace{-20pt}
      \begin{minipage}[h]{1\linewidth}
          \centering
          \includegraphics[scale=0.48]{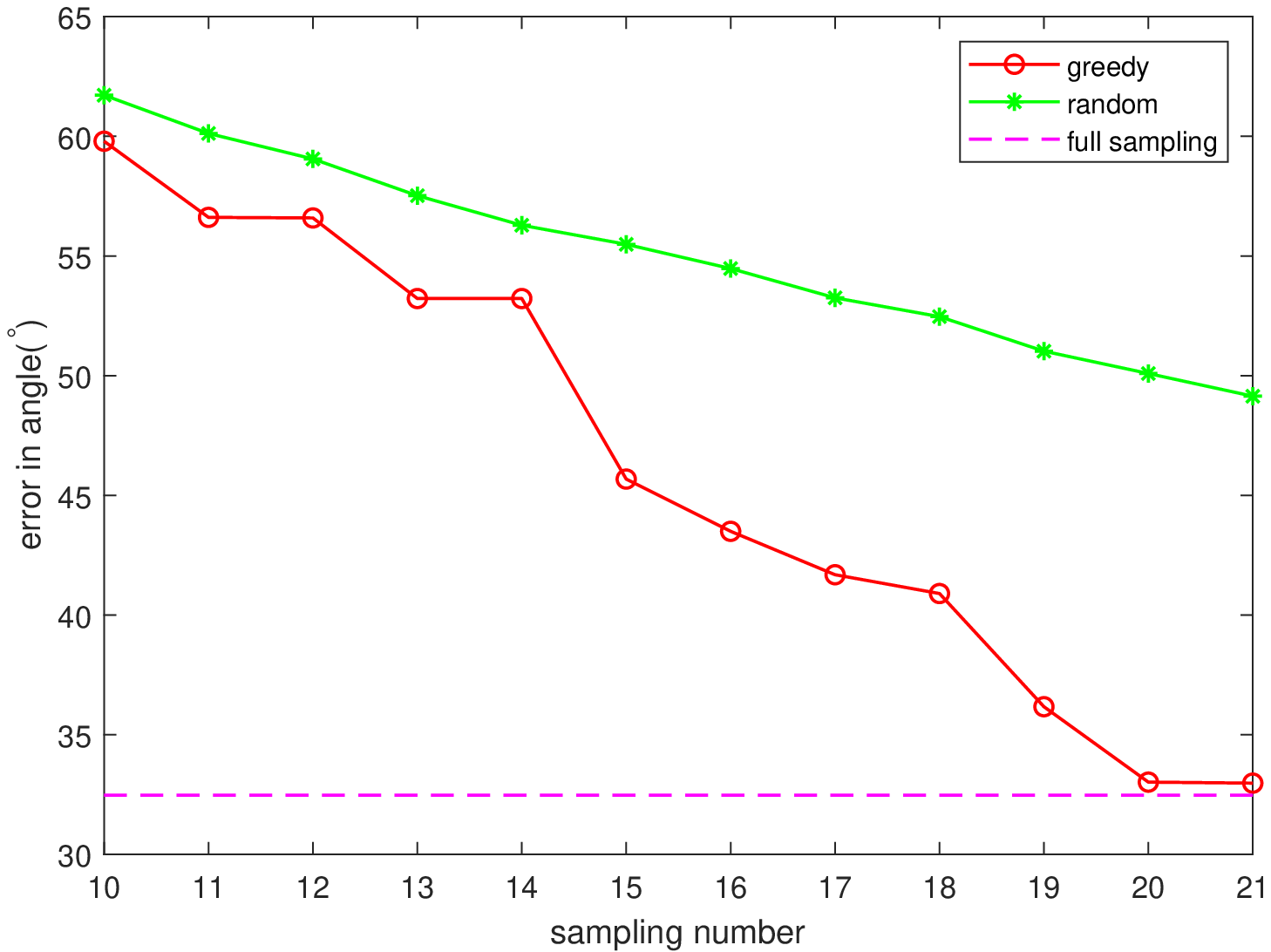}
      \end{minipage}
  }
  \subfigure[]
  {
	 \label{iter}
	 \hspace{-20pt}
     \begin{minipage}[h]{1\linewidth}
          \centering
          \includegraphics[scale=0.48]{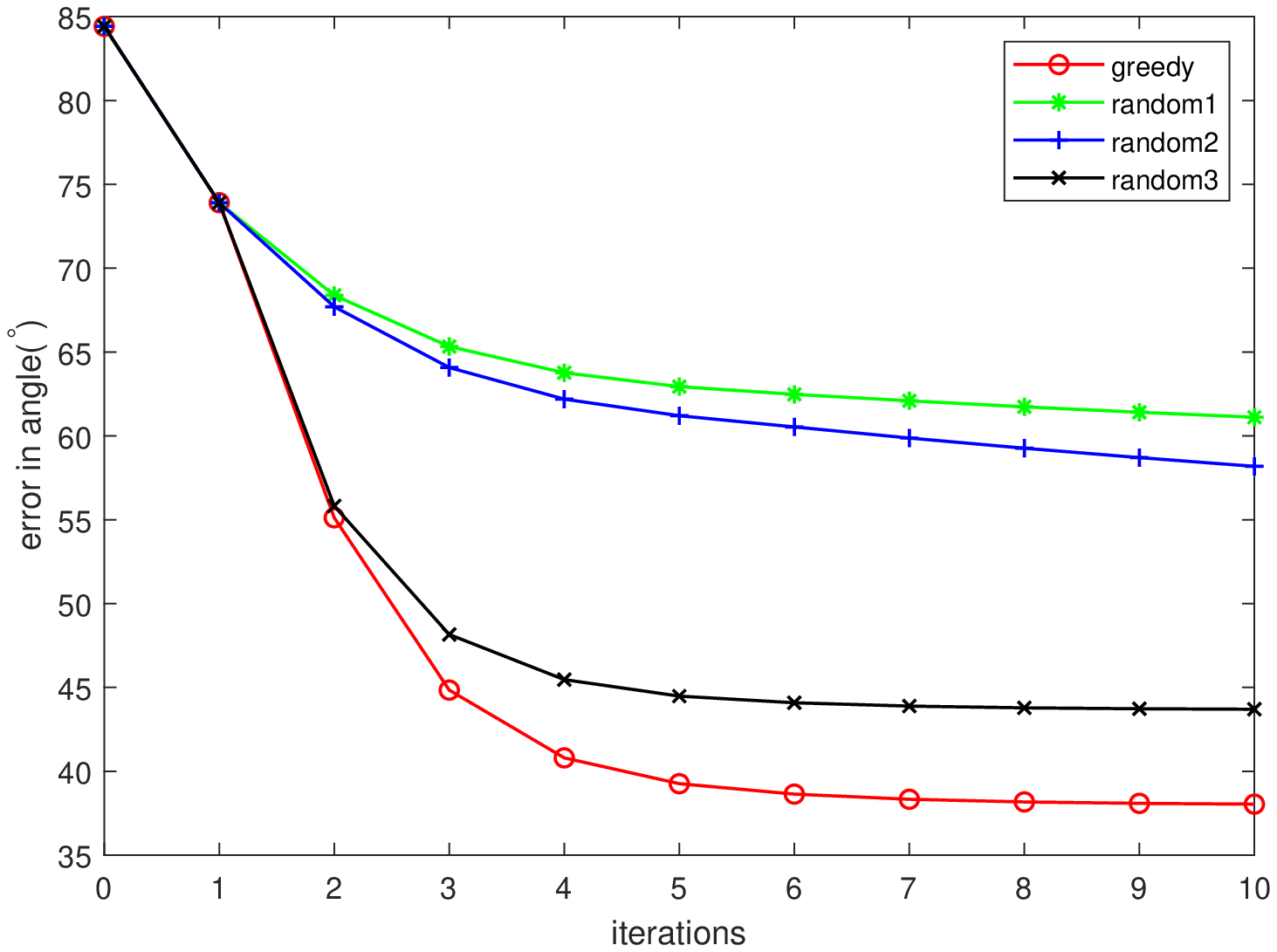}
      \end{minipage}
  }
  \caption{Comparison of recovery performance between greedy sampling and random sampling (a) at different sampling rate (b) during a iteration at 40\% sampling rate}
  \label{comparsion}
\end{figure}

As we can see, with the increase of sampling rate, the average error in angle goes down continuously. Among them, the error in angle of the greedy sample set is always the smallest, 
i.e. the recovery performance is the best.

In addition, we give the trend of the error in angle during the iteration in a certain recovery process at 40\% sampling rate in Figure \ref{comparsion}\subref{iter}. From the figure, 
under the proposed recovery algorithm, the error in angle converges quickly. Besides, compared with random sampling, the error in angle of greedy sampling set descends more steeply.

\section{Conclusion}
\label{sec:conclusion}
In this paper, we put forward the idea of sign sampling and recovery of band-limited graph signals.
Based on sign information on vertexes, we propose corresponding reconstruction and greedy sampling algorithm for band-limited graph siganl.
On the one hand, we guarantee the performance of the recovery algorithm by theoretical proof. On the other hand, simulation results show that greedy sampling algorithm
has better performance than random sampling.
In all, we think this work has practical significance to a certain extent and could be valuable in the field of sampling and recovery of graph signals.

\section{Appendix}
\label{sec:appendix}

\subsection{Appendix A}
\begin{lemma}
  Each projection is a firmly non-expansive operation.  
\end{lemma}
\begin{proof}
  From (\ref{Pb}), we can get $\left\|2\boldsymbol{P}_b-\boldsymbol{I}\right\|=\|2\boldsymbol{U}\boldsymbol{\Gamma}\boldsymbol{U}^T-\boldsymbol{I}\|$.
  $\boldsymbol{P}_b$ has only 2 eigenvalues, $0$ and $1$, because the diagonal elements of $\Gamma$ are either $0$ or $1$. Then   
  $2\boldsymbol{U}\boldsymbol{\Gamma}\boldsymbol{U}^T-\boldsymbol{I}$ has only 2 eigenvalues, $-1$ and $1$.
  It is not hard to demonstrate that $2\boldsymbol{U}\boldsymbol{\Gamma}\boldsymbol{U}^T-\boldsymbol{I}$ is symmetric, so 
  $(2\boldsymbol{U}\boldsymbol{\Gamma}\boldsymbol{U}^T-\boldsymbol{I})^T(2\boldsymbol{U}\boldsymbol{\Gamma}\boldsymbol{U}^T-\boldsymbol{I})=(2\boldsymbol{U}\boldsymbol{\Gamma}\boldsymbol{U}^T-\boldsymbol{I})^2$, which
  only has an eigenvalue of $1$. Hence, $\left\|2\boldsymbol{P}_b-\boldsymbol{I}\right\|=1$, the operator 
  $2\boldsymbol{P}_b-\boldsymbol{I}$ is non-expansive, i.e., $\boldsymbol{P}_b$ is firmly non-expansive.

  From (\ref{Pv}), we can get 
  \begin{equation*}
    \left(\left(2\boldsymbol{P}_v-\boldsymbol{I}\right)\boldsymbol{x}\right)_j=
    \left\{
      \begin{array}{lr}
        x_j      & sign(x_j) = \left(\boldsymbol{S}_x\right)_i \\
        -x_j                 & sign(x_j) \not = \left(\boldsymbol{S}_x\right)_i
      \end{array}
    \right.
  \end{equation*}
  Hence, $\left\|\left(2\boldsymbol{P}_v-\boldsymbol{I}\right)\boldsymbol{x}\right\|=\left\|\boldsymbol{x}\right\|$, the operator 
  $2\boldsymbol{P}_v-\boldsymbol{I}$ is non-expansive, i.e., $\boldsymbol{P}_v$ is firmly non-expansive.

\end{proof}

\begin{lemma}
  $\{C_b,C_v\}$ is boundedly linearly regular.
\end{lemma}
\begin{proof}
  For $C_b$,$C_v$ and $C$, they are all closed convex cones in finite-dimensional Hilbert spaces.
  By means of Proposition 5.4 and 5.9 in \cite{bauschke1996projection}, we can derive that the $\{C_b,C_v\}$ is boundedly linearly regular.
\end{proof}

  In (\ref{iterformula}), Each iteration is actually made up of two firmly non-expansive operations.
  According to the algorithm settings in \cite{bauschke1996projection}, it is not difficult to verify that this algorithm is cyclic, singular, and unrelaxed.
  Furthermore, refer to Definition 3.18 in \cite{bauschke1996projection}, we can know that the sequence of projection operators converges actively pointwise, meanwhile the iterative sequence $\{\boldsymbol{x}_n\}$ is asymptotically regular
  and Fej\'er monotone in terms of Corollary 3.5 and Lemma 3.2(iv) in \cite{bauschke1996projection}, which yields $\{\boldsymbol{x}_n\}$ converges in norm.

  With all these discoveries, under the help of lemma 2 and using Corollary 3.12 in \cite{bauschke1996projection}, it is proved that $\left\{\boldsymbol{x}_\mathit{n}\right\}$ converges to some point in $C$.
  Then by Theorem 5.7 in \cite{bauschke1996projection}, we can futher conclude that the conver-gence rate is independent of the selection of the initial point $\boldsymbol{x}_0$.
\label{A-A}

\bibliographystyle{IEEEbib}
\bibliography{ref}

\begin{thebibliography}{10}

\bibitem{stankovic2019understanding}
Ljubisa Stankovic, Danilo~P Mandic, Milos Dakovic, Ilia Kisil, Ervin Sejdic,
  and Anthony~G Constantinides,
\newblock ``Understanding the basis of graph signal processing via an intuitive
  example-driven approach [lecture notes],''
\newblock {\em IEEE Signal Processing Magazine}, vol. 36, no. 6, pp. 133--145,
  2019.

\bibitem{zhou2004regularization}
Dengyong Zhou and Bernhard Sch{\"o}lkopf,
\newblock ``A regularization framework for learning from graph data,''
\newblock in {\em ICML 2004 Workshop on Statistical Relational Learning and Its
  Connections to Other Fields (SRL 2004)}, 2004, pp. 132--137.

\bibitem{zhang2020deep}
Ziwei Zhang, Peng Cui, and Wenwu Zhu,
\newblock ``Deep learning on graphs: A survey,''
\newblock {\em IEEE Transactions on Knowledge and Data Engineering}, 2020.

\bibitem{shuman2013emerging}
David~I Shuman, Sunil~K Narang, Pascal Frossard, Antonio Ortega, and Pierre
  Vandergheynst,
\newblock ``The emerging field of signal processing on graphs: Extending
  high-dimensional data analysis to networks and other irregular domains,''
\newblock {\em IEEE signal processing magazine}, vol. 30, no. 3, pp. 83--98,
  2013.

\bibitem{stankovic2020data}
Ljubi{\v{s}}a Stankovi{\'c}, Danilo Mandic, Milo{\v{s}} Dakovi{\'c},
  Milo{\v{s}} Brajovi{\'c}, Bruno Scalzo, Shengxi Li, Anthony~G Constantinides,
  et~al.,
\newblock ``Data analytics on graphs part ii: Signals on graphs,''
\newblock {\em Foundations and Trends{\textregistered} in Machine Learning},
  vol. 13, no. 2-3, 2020.

\bibitem{xia2021graph}
Feng Xia, Ke~Sun, Shuo Yu, Abdul Aziz, Liangtian Wan, Shirui Pan, and Huan Liu,
\newblock ``Graph learning: A survey,''
\newblock {\em arXiv preprint arXiv:2105.00696}, 2021.

\bibitem{cvetkovic2009applications}
Drago{\v{s}}~M Cvetkovic,
\newblock ``Applications of graph spectra: An introduction to the literature,''
\newblock {\em Appl. Graph Spectra}, vol. 13, no. 21, pp. 7--31, 2009.

\bibitem{anis2016efficient}
Aamir Anis, Akshay Gadde, and Antonio Ortega,
\newblock ``Efficient sampling set selection for bandlimited graph signals
  using graph spectral proxies,''
\newblock {\em IEEE Transactions on Signal Processing}, vol. 64, no. 14, pp.
  3775--3789, 2016.

\bibitem{tanaka2020sampling}
Yuichi Tanaka, Yonina~C Eldar, Antonio Ortega, and Gene Cheung,
\newblock ``Sampling signals on graphs: From theory to applications,''
\newblock {\em IEEE Signal Processing Magazine}, vol. 37, no. 6, pp. 14--30,
  2020.

\bibitem{girault2020graph}
Benjamin Girault, Antonio Ortega, and Shrikanth~S Narayayan,
\newblock ``Graph vertex sampling with arbitrary graph signal hilbert spaces,''
\newblock in {\em ICASSP 2020-2020 IEEE International Conference on Acoustics,
  Speech and Signal Processing (ICASSP)}. IEEE, 2020, pp. 5670--5674.

\bibitem{tsitsvero2016signals}
Mikhail Tsitsvero, Sergio Barbarossa, and Paolo Di~Lorenzo,
\newblock ``Signals on graphs: Uncertainty principle and sampling,''
\newblock {\em IEEE Transactions on Signal Processing}, vol. 64, no. 18, pp.
  4845--4860, 2016.

\bibitem{jacques2013robust}
Laurent Jacques, Jason~N Laska, Petros~T Boufounos, and Richard~G Baraniuk,
\newblock ``Robust 1-bit compressive sensing via binary stable embeddings of
  sparse vectors,''
\newblock {\em IEEE Transactions on Information Theory}, vol. 59, no. 4, pp.
  2082--2102, 2013.

\bibitem{boyd2004convex}
Stephen Boyd, Stephen~P Boyd, and Lieven Vandenberghe,
\newblock {\em Convex optimization},
\newblock Cambridge university press, 2004.

\bibitem{theodoridis2010adaptive}
Sergios Theodoridis, Konstantinos Slavakis, and Isao Yamada,
\newblock ``Adaptive learning in a world of projections,''
\newblock {\em IEEE Signal Processing Magazine}, vol. 28, no. 1, pp. 97--123,
  2010.

\bibitem{barker1973lattice}
George~Phillip Barker,
\newblock ``The lattice of faces of a finite dimensional cone,''
\newblock {\em Linear Algebra and its Applications}, vol. 7, no. 1, pp. 71--82,
  1973.

\bibitem{perraudin2014gspbox}
Nathana{\"e}l Perraudin, Johan Paratte, David Shuman, Lionel Martin, Vassilis
  Kalofolias, Pierre Vandergheynst, and David~K Hammond,
\newblock ``Gspbox: A toolbox for signal processing on graphs,''
\newblock {\em arXiv preprint arXiv:1408.5781}, 2014.

\bibitem{bauschke1996projection}
Heinz~H Bauschke and Jonathan~M Borwein,
\newblock ``On projection algorithms for solving convex feasibility problems,''
\newblock {\em SIAM review}, vol. 38, no. 3, pp. 367--426, 1996.

\end{thebibliography}

\end{document}